\documentclass{llncs}
\pdfoutput=1 
\usepackage{algorithmicx}
\usepackage{algorithm}
\usepackage{algpseudocode}
\usepackage{changebar}
\usepackage{graphicx}
\usepackage[english]{babel}

\usepackage{amsthm}
\usepackage{amsmath,amssymb}
\usepackage{xcolor}
\usepackage[thicklines]{cancel}
\usepackage{stfloats} 
\usepackage{subfigure}
\usepackage{eqparbox}
\usepackage{float}
\usepackage{graphicx,url}
\usepackage{todonotes}
\usepackage{version}
\usepackage{setspace}	
\usepackage{array}
\usepackage{pslatex}
\usepackage{color}
\usepackage{epstopdf}
\usepackage{comment}
\usepackage{cite}
\usepackage{apalike}
\usepackage{fixltx2e}
\usepackage{pgfplots}
\usepackage{tikz}
\usetikzlibrary{arrows, positioning, calc}
\newtheorem*{theorem*}{Theorem}

\newtheorem*{partialsuppression}{Partial Suppression}

\begin{document}
	
	\title{Safe Disassociation of Set-Valued Datasets}
		
	\author{Nancy Awad\inst{1,2}\and 
		Bechara AL Bouna\inst{1} \and 
		Jean-Francois Couchot\inst{2} \and 
		Laurent Philippe\inst{2}}
	
	\institute{
		TICKET Lab., Antonine University, Hadat-Baabda, Lebanon.\\
		\email{nancy.awad,bechara.albouna@ua.edu.lb}
		\and
		FEMTO-ST Institute, UMR 6174 CNRS, Universit\'{e} of Bourgogne Franche-Comt\'e, France. \\
		\email{jean-francois.couchot, laurent.philippe@univ-fcomte.fr}
	}

	\maketitle 
	
	%ABSTRACT===========================================================================
	\abstract{
		
		Disassociation introduced by Terrovitis \textit{et al.} is a bucketization based anonimyzation technique that divides a set-valued dataset into several clusters to hide the link between individuals and their complete set of items. It increases the utility of the anonymized dataset, but on the other side, it raises many privacy concerns, one in particular, is when the items are tightly coupled to form what is called, a cover problem. 
		In this paper, we present safe disassociation, a technique that relies on partial-suppression, to overcome the aforementioned privacy breach encountered when disassociating set-valued datasets. Safe disassociation allows the $k^m$-anonymity privacy constraint to be extended to a bucketized dataset and copes with the cover problem. We describe our algorithm that achieves the safe disassociation and we provide a set of experiments to demonstrate its efficiency. 
	}
	
	\section{Introduction} \label{sec:introduction}
	Privacy preservation is a key concern in data publishing where individual's personal information must remain protected under all circumstances. This sounds straightforward, but it is somehow difficult to achieve. The AOL search data leak in 2006 \cite{aol2006} is an explicit example that shows the consequences of a unsupervised data publishing. The query logs of 650k individuals were released after omitting all explicit identifiers. They were later withdrawn due to multiple reports of attackers linking individuals to their sensitive records. Alternatively, providing a "complete" privacy over the data requires sacrifices in terms of utility, in other words, usefulness of the data \cite{kanon_defn,kanon_sweeney,ldiversity,anatomy,Dwork2006}.  
Hence, it is pointless to publish datasets that do not provide valuable information. A suitable trade-off between data utility and privacy must be achieved. The point is to provide not only a value anonymization technique, but instead a dataset anonymization technique, that hides/anonymizes the link between individuals and their sensitive information, and, at the same time, keeps the dataset useful for analysis.
When publishing a set-valued dataset (e.g., shopping and search items) it is important to pay attention to attackers that try intentionally to link individuals to their sensitive information. These attackers may be able to single out an individual's complete record/itemset by associating data items from the dataset to their background knowledge.
The example in Figure~\ref{fig:covpb:original} shows a set-valued dataset $\mathcal{T}$ consisting of 6 records $r_1$,\ldots,$r_6$, which are itemsets linked to individuals $1$, \ldots, $6$ respectively. For instance, $r_1:\{a, e\}$ can be interpreted as individual $1$ has searched for item $a$ and item $e$. If an attacker knows that individual $1$ has searched for items $d$ and $e$, he/she will be able to link $1$ to his record $r_2$.

\begin{figure}[ht]
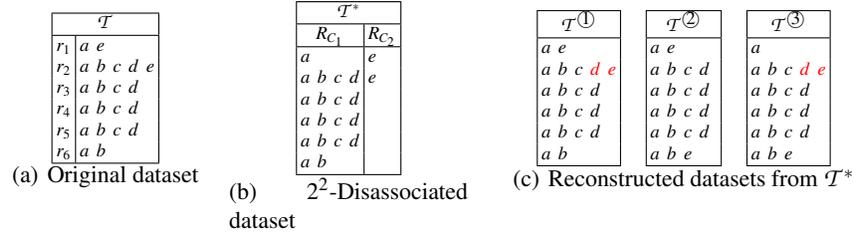

	\begin{center}
		\begin{scriptsize}
			\subfigure[Original dataset]{
				\begin{minipage}{0.25\textwidth}
					\begin{center}
						\[
						\begin{array}{|l|lllll|}
						\hline
						\multicolumn{6}{|c|}{\mathcal{T}}\\
						\hline
						r_1 & a & e &   &   &  \\
						r_2 & a & b & c & d & e \\
						r_3 & a & b & c & d & \\
						r_4 & a & b & c & d & \\
						r_5 & a & b & c & d & \\
						r_6 & a & b &   &   & \\
						\hline
						\end{array}
						\]  
					\end{center}
				\end{minipage}
				\label{fig:covpb:original}
			}
			\subfigure[$2^2$-Disassociated dataset]{
				\begin{minipage}{0.25\textwidth}
					\begin{center}
						\[
						\begin{array}{|llll|l|}
						\hline
						\multicolumn{5}{|c|}{\mathcal{T}^*}\\
						\hline
						\multicolumn{4}{|c|}{R_{C_1}} & {R_{C_2}} \\
						\hline
						a &   &   &   & e \\
						a & b & c & d & e  \\
						a & b & c & d & \\
						a & b & c & d & \\
						a & b & c & d & \\
						a & b &   &   & \\
						\hline
						\end{array}
						\]  
					\end{center}
				\end{minipage}
				\label{fig:covpb:disassosiated}
			}   
			\subfigure[Reconstructed datasets from $\mathcal{T}^*$]{
				\begin{minipage}{0.45\textwidth}
					\begin{center}
						\[
						\begin{array}{|lllll|}
						\hline
						\multicolumn{5}{|c|}{\mathcal{T}^{\textcircled{1}}}\\
						\hline
						a & e &   &   &  \\
						a & b & c & {\color{red}d} & {\color{red}e} \\
						a & b & c & d & \\
						a & b & c & d & \\
						a & b & c & d & \\
						a & b &   &   & \\
						\hline
						\end{array}        
						\,\,\,\,\,\,\,\,\,\,\,\,
						\begin{array}{|lllll|}
						\hline
						\multicolumn{5}{|c|}{\mathcal{T}^{\textcircled{2}}}\\
						\hline
						a & e &   &   &  \\
						a & b & c & d &  \\
						a & b & c & d & \\
						a & b & c & d & \\
						a & b & c & d & \\
						a & b & e &   & \\
						\hline
						\end{array}        
						\,\,\,\,\,\,\,\,\,\,\,\,
						\begin{array}{|lllll|}
						\hline
						\multicolumn{5}{|c|}{\mathcal{T}^{\textcircled{3}}}\\
						\hline
						a &   &   &   &  \\
						a & b & c & {\color{red}d} & {\color{red}e} \\
						a & b & c & d & \\
						a & b & c & d & \\
						a & b & c & d & \\
						a & b & e &   & \\
						\hline
						\end{array}       
						\]  
					\end{center}
				\end{minipage}
				\label{fig:covpb:preimages}
			}
		\end{scriptsize}
	\end{center}
	\label{fig:pbreach}
	\caption{Disassociation leading to a cover problem}
\end{figure}

Several techniques \cite{kanon_defn,kanon_sweeney,anatomy,slicing,sabrina2013,Wang2016,wang2017} have been defined in the literature to anonymize the dataset and cope with this particular association problem.
Anonymization by disassociation~\cite{terrovitis08,loukidesLGT14,Loukides2015} is a bucketisation technique \cite{anatomy,slicing,sabrina2013} that keeps the items without alternation/generalization but separates the records into clusters and chunks to hide their associations. 
More specifically,  disassociation transforms the original data into 
$k^m$-anonymous clusters of chunks to ensure that an attacker who knows up to $m$ items cannot associate them with less than $k$ records.
First, disassociation divides the dataset horizontally, creating clusters of records with similar frequent items. In Figure~\ref{fig:covpb:original}, the item $a$ has the highest number of occurrences in the records; thus, all records containing the item $a$ are added to the cluster. Next, disassociation divides the clusters vertically, creating inside each cluster $k^m$-anonymous record chunks of items, and one item chunk containing items that appear less than $k$ times.
Figure~\ref{fig:covpb:disassosiated} shows the result of disassociation with only one cluster, namely $\mathcal{T}^*$, containing 2 record chunks $R_{C_1}$, and $R_{C_2}$ and no item chunk. Both $R_{C_1}$ and $R_{C_2}$ are $2^2$-anonymous.
That is, an attacker who knows any two items about an individual will be able link them to at least two records from the dataset.

In a previous work~\cite{DBLP:conf/secrypt/BarakatBNG16}, the disassociation technique has been evaluated and found to be vulnerable to what is called, a cover problem. This cover problem provides %strong, moderate, and weak
attackers with the ability to associate itemsets in consequent record chunks with individuals' records, and compromises the disassociated dataset.
Figure~\ref{fig:covpb:preimages} highlights all datasets 
$\mathcal{T}^{\textcircled{1}}$, $\mathcal{T}^{\textcircled{2}}$, and $\mathcal{T}^{\textcircled{3}}$ 
that could be reconstructed from the disassociated dataset $\mathcal{T}^*$ by associating records from different record chunks.
Suppose now the attacker knows that an individual has searched for items $d$ and $e$ , $\{d,e\}$. In such a case, he/she can remove $\mathcal{T}^{\textcircled{2}}$ from the possible reconstructed dataset and will be able to link every record containing $d$ and $e$ with certainty to $\{a,b,c\}$ from the two remaining reconstructed datasets leading to a privacy breach.

This paper extends the previous work~\cite{DBLP:conf/secrypt/BarakatBNG16} by providing a theoretical and practical solution to this cover problem. 
It introduces a privacy preserving technique to safely disassociate a set-valued dataset and address this cover problem.
This method is further denoted as safe disassociation.
Our contributions are summarized as follows:
\begin{itemize}
\item We define the privacy guarantee for safe disassociation and provide the appropriate algorithm to achieve it.
%using partial suppression.
\item We investigate the efficiency of safe disassociation and its impact on the utility of aggregate analysis and the discovery of association rules.
\end{itemize}

The rest of the paper is organized as follows. 
Section 2 presents an overview of some of the related works in set-valued dataset anonymization .
%First of all, we recall the cover problem, already defined, and investigate its repercussions 
%on the privacy of a disassociated dataset.
In Section~\ref{sec:formalization}, 
we describe the formal model of disassociation and discuss the cover problem that makes disassociation vulnerable.
We define safe disassociation in Section.~\ref{sec:PPPS}. 
Experimental evaluations of this method is presented in Section~\ref{sec:exper}.
Finally, we conclude in Section~\ref{sec:conclusion} and present outlines for future work.

	\section{Related Work} \label{sec:relwork}
	%Data is becoming the new currency where analysis over the data are shiffting our world in the marketing field, the political campaigns, the medical sector and thousands of other domains. However, there exists another aspect of this currency that reveals a threat related to the privacy of the individuals linked to the data. Whether those data are treated in a very ethical environment or not, the need for privacy protecting algorithms is crucial and requested by many laws. 
Anonymization techniques can be divided into several categories, namely categorization, generalization, %modification, 
and bucketization, detailed hereafter.
	
With categorization~\cite{Jia2014,Chen2016,wang2017}, attributes are classified into several categories with respect to their sensitivity: identifying, sensitive or non-sensitive. 
	%In , the authors propose a privacy constraint $\rho$-uncertainty, to ensure that there are no sensitive association rules that can be inferred with confidence larger than $\rho$. 
On the opposite, all the data in this current work are considered to be grouped into sets whose items have the same level of sensitivity; the items combined together are sensitive and therefore only the link between the individual and their corresponding items must be protected.

Generalization techniques ~\cite{kanon_defn,kanon_sweeney,ldiversity} create homogeneous subsets by replacing values with ranges wide enough to create ambiguity. 
Reducing the extent of generalization, and thus decreasing the information loss is critical to preserving the usefulness of the data. 
In~\cite{yeye2009}, only a local generalization is applied whereas in~\cite{fard2010effective} a clustering based technique
is implemented to minimize the abstraction.
On the opposite, initial data in this current work are not generalized to ensure privacy or even modified.

Bucketization techniques~\cite{anatomy,fragmentation,fragmentationInference} 
are valuable due to their ability to keep the values intact.
Identifiable links between items are hidden by separating the attributes of the data. 
%These techniques are interesting when studying the utility of the data. 
Under this category lies the disassociation technique~\cite{terrovitis2012,loukidesLGT14,Loukides2015,bewong2017} that works on clustering the data and hiding identifiable links in each cluster by separating the attributes. % in each cluster. 	
Unfortunately, disassociation has shown to be vulnerable to a privacy breach when the items are tightly coupled. It is the ability of an attacker to link his/her partial background knowledge represented by at most $m$ items that he/she is allowed to have, according to the privacy constraint $k^m$-anonymity, with certainty to less than $k$ distinct records in a disassociated dataset.  
	
Differential privacy ~\cite{Dwork2006} is an anonymization technique that adds noise to the query results. It is based on a strong mathematical foundation that guarantees that an attacker is unable to identify sensitive data about an individual if his/her information were removed from the dataset.
Unlike differential privacy, this disassociation based  work does not distinguish between sensitive and non-sensitive attributes and is capable to retrieve viable and trustful information that has not been altered nor modified.
The authors in \cite{Zhang2015} defined cocktail a framework that uses both disassociation to anonymize the data for publishing as well as differential privacy to add noise on data querying. In spite of the originality of the idea, their technique subsumes the drawbacks of both disassociation and differential privacy; as it is vulnerable to the cover problem and releases questionable information.

	\section{Background} \label{sec:formalization}

To be self-content, this section recalls the basis of disassociation (Sec.~\ref{sub:model} and Sec.~\ref{sec:disass}) introduced in~\cite{terrovitis2012}. 
Next (Sec.~\ref{sec:pb}), it exhibits a class of privacy breach called cover problem inherent to this anonymization technique already sketched  in \cite{DBLP:conf/secrypt/BarakatBNG16} showing its repercussion on data privacy.

\subsection{Data model}\label{sub:model}
Let $\mathcal{D} = \{x_1,..., x_d\}$ be a set of items (\textit{e.g.}, supermarket products, query logs, or search keywords). Any subset $I \subseteq \mathcal{D}$ is an itemset (\textit{e.g.}, items searched together). Let $\mathcal{T}=\{r_1,..., r_n\}$ be a dataset of records where 
each $r_i \subseteq \mathcal{D}$ for $1\leq i \leq n$ is a record %in $\mathcal{T}$ 
and $r_i$ is associated with a specific individual $i$ of a population.
%Note that $\mathcal{D}$ is no more than the set of all the items in $\mathcal{T}$: $\mathcal{D}=\bigcup^{n}_{i=1}r_i$.
Let $R_\mathcal{T}$ be a subset of records in $\mathcal{T}$. 
Both $\mathcal{T}$ and $R$ have the multiset semantic, which
can contain more than one instance for each of its elements. 

With such notations, $s(I,\mathcal{T})$ is the number of records in $\mathcal{T}$ that contain all the elements in $I$.
More formally, it is defined by the following equation
\begin{equation}
    s(I,\mathcal{T})= \left\vert \left\{ r \in \mathcal{T} \mid I \subseteq r \right\}\right\vert
    \label{eq:support}
\end{equation}
By extension $s(\mathcal{T}) = s(\emptyset, \mathcal{T})$ and $s(R_T) = s(\emptyset, R_T)$ are the number of records in $\mathcal{T}$ and $R_T$ respectively. 

Table \ref{tab:notations} recalls the basic concepts and notations used in the paper. 
\begin{table}
\centering
\caption{Notations used in the paper}
\begin{scriptsize}
\begin{tabular}{|l|p{8cm}|}\hline
$\mathcal{D}$ & a set of items \\ \hline
$\mathcal{T}$ & a dataset containing individuals related records \\ \hline
$\mathcal{T}^*$ & a disassociated dataset i.e. a dataset anonymized using the disassociation technique\\ \hline
$\mathcal{T}^{\textcircled{.}}$ & a dataset reconstructed by cross joining the itemsets of the record chunks in a cluster from the disassociated dataset\\ \hline
$r$ &  a record (of $\mathcal{T}$) which is set of items associated with a specific individual of a population \\ \hline
$I$ &  an itemset included in $\mathcal{D}$ %that might or not be grouped together in a record of $\mathcal{T}$
\\ \hline
$s(I, \mathcal{T})$ & support of $I$ in $\mathcal{T}$ i.e. the number of records in $\mathcal{T}$ that are superset of $I$\\ \hline
$R$ &  a cluster in a disassociated dataset, formed by the horizontal partitioning of $\mathcal{T}$ \\ \hline
%$R_T$ &  a subset of records in  $\mathcal{T}$ \\ \hline
%$C$ &  a chunk in a disassociated dataset; is a set of sub-records of $\mathcal{T}$ in $\mathcal{T}^*$  \\ \hline
$R_T$ & an item chunk in a disassociated cluster\\ \hline
$R_C$ & a record chunk in a disassociated cluster  \\ \hline
$\delta$ &  maximum number of records allowed in a cluster \\ \hline
$n$ &  number of records in $\mathcal{T}$ \\ \hline
\end{tabular}
\end{scriptsize}
\label{tab:notations}
\end{table}

\subsection{Disassociation} \label{sec:disass}

Disassociation works under the assumption that the items should neither be altered, suppressed, nor generalized,  but at the same time the resulting dataset must respect the $k^m$-anonymity privacy constraint \cite{terrovitis08}.
Formally, $k^m$-anonymity is defined as follows:

\begin{definition}[$k^m$-anonymity] Given a dataset of records $\mathcal{T}$ whose items belong to a given set of items $\mathcal{D}$. The dataset $\mathcal{T}$ is $k^m$-anonymous if $\forall I \subseteq \mathcal{D}$ such that $|I|\leq m$, the number of records in $\mathcal{T}$ that are superset of $I$  is greater than or equal to $k$, \textit{i.e.}, $s(I, \mathcal{T}) \geq k$.

\end{definition}

%The disassociation technique \cite{terrovitis2012,Loukides2014}, for its part, ensures privacy through bucketization to provide better utility when it comes to frequent itemsets discovery and aggregate analysis. 

%Disassociation separates the dataset into clusters of $k^m$-anonymous record chunks and an item chunk. The key idea is to sort records based on the most frequent items and then group them horizontally into smaller disjoint clusters $\{R_1$, \dots, $R_q\}$. In a next step, each cluster $R_i$ is 
%vertically partitioned  into $k^m$-anonymous chunks $\{C_1,...,C_{t}\}$ and an item chunk $C_{T}$ to hide infrequent combinations. These chunks are created subsequently as long as there are items that can be grouped together in a way to satisfy the $k^m$-anonymity privacy constraint. The remaining items, the ones that occur less than $k$ times, are moved to the item chunk $C_{T}$\footnote{The shared chunk, as defined in the original paper \cite{terrovitis2012}, is omitted here for simplicity.}.

Given a dataset $\mathcal{T}$, applying $k^m$-disassociation\footnote{In what follows, we use $k^m$-disassociation to denote a dataset that is disassociated and satisfies $k^m$-anonymity.} on $\mathcal{T}$ produces a dataset $\mathcal{T}^*$ composed of $q$ clusters, each divided into a set of record chunks and an item chunk,

\begin{small}
\begin{equation*}
\mathcal{T}^*=\left\{
\{R_{1_{C_1}},\dots, R_{1_{C_{t}}},R_{1_Y}\} , \dots,
\{R_{q_{C_1}},\dots, R_{q_{C_{v}}},R_{q_T} \} \right\}
\end{equation*}
\end{small}

such that $\forall R_{i_{C_j}} \in \mathcal{T}^*$, $R_{i_{C_j}}$ is $k$-anonymous,
where,
\begin{itemize}
\item $R_{i_{C_j}}$ represents the itemsets of the $i^{th}$ cluster that are contained in its $j^{th}$ record chunk. 
\item $R_{i_T}$ is the item chunk of the $i^{th}$ cluster containing items that occur less than $k$ times.
\end{itemize}

The example in Figure \ref{fig:covpb:disassosiated} shows that the $2^2$-disassociated dataset contains only one cluster with two $2^2$-anonymous record chunks. 
We thus have $\mathcal{T}^*=\left\{R_{C_{1}}, R_{C_{2}}\right\}$ with
$R_{C_{1}} = \{ \{a\}, \{a,b,c,d\}, \{a,b,c,d\}, \{a,b,c,d\},\{a,b,c,d\},\{a,b\}\}$, and 
$R_{C_{2}} = \{ \{e\}, \{e\}\}$. 

According to ~\cite{terrovitis2012} and by construction, $k^m$-disassociation guarantees all the produced record chunks are  $k^m$-anonymous. This can be better explained in this example as: any combination of two items ($m=2$) from Figure \ref{fig:covpb:original}, for example $\{a,b\}$, is found at least in two records ($k=2$) in the record chunk, thus satisfying $k^m-anonymity$ in Figure \ref{fig:covpb:disassosiated}.\\
However, to ensure the privacy of a disassociated dataset, $k^m$-anonymity has to be guaranteed in one of the valid reconstructed datasets of $\mathcal{T}^*$ since an attacker can produce all of them, provided he/she knows $\mathcal{T}^*$, $k$, and $m$ \cite{terrovitis2012}. 
This privacy guarantee is formally expressed as follows:

\begin{definition}[Disassociation Guarantee]\label{def:disassociation:guarantee}
Let $\mathcal{G}$ be the inverse transformation of $\mathcal{T}^*$ with respect to a $k^m$ disassociation, \textit{i.e.},
the set of all possible datasets whose $k^m$ disassociation would yield $\mathcal{T}^*$.
\emph{Disassociation guarantee} is established if for any $I \subseteq \mathcal{D}$ such that  
$|I| \leq m$,  there exists  $\mathcal{T}^{\textcircled{.}} \in \mathcal{G}(\mathcal{T}^*)$ with 
$s(I, \mathcal{T}^{\textcircled{.}}) \geq k$.  
\end{definition} 

The Disassociation Guarantee ensures that for any individual with a complete record $r$, and for an attacker who knows up to $m$ items of $r$, at least one of the datasets reconstructed by the inverse transformation contains the record $r$ $k$ times or more. 
That is, the record $r$, as all other records, exists $k$ times in at least one of the inverse transformations.    

The authors in \cite{DBLP:conf/secrypt/BarakatBNG16} demonstrate that this disassociation guarantee is not enough to ensure privacy. They show that whenever a disassociated dataset is subject to cover problem, a privacy breach might be encountered. In the following section, we briefly present the cover problem.

\subsection{Cover Problem}\label{sec:pb}
Let us recall, from Subsection~\ref{sec:disass} that the disassociation technique hides itemsets that occur less than $k$ times in the original dataset for a given $m$ items, by 1) dividing them into $k^m$-anonymous sub-records in record chunks and 2) ensuring that all the records reconstructed by the inverse transformation are $k^m$-anonymous in at least one of the resulting datasets.

A cover problem is defined by the ability to associate one-to-one or one-to-many items in two distinct record chunks, from the same cluster, in the disassociated data. Without loss of generality, we focus on one cluster $R=\{R_{C_1},\ldots, R_{C_t},R_{T} \}$ of $\mathcal{T}^*$ resulting from a $k^m$-disassociation.
%The target item from the first chunk has equal or higher support and we call it the covered item.  
Formally, the cover problem is defined as follows.  

\begin{definition}[Cover Problem]\label{def:coverpb}
	Let $I_j$ be the set of items in $R_{C_j}$,
	$I_{j} = \{x \in  R_{C_{j}}\}$.
	If there exists an item $z\in I_j$ such that the support of $I_j$ is equal to the support of the singleton $\{z\}$ in $R_{C_{j}}$, \textit{i.e.},
	\begin{equation}\label{eq:coverpb}
	s(I_j, R_{C_{j}})  = s(\{z\}, R_{C_{j}}),
	\end{equation}	
	the cluster $R$, and the dataset $\mathcal{T}^*$ as a consequence, are subject to a \emph{cover problem}.
    
    \end{definition}
$\forall$ $z \in I_j$, if $z$ satisfies equation~\eqref{eq:coverpb}, $z$ is denoted as \emph{covered item}. The set of all the covered items is denoted as $L_j$, which is contained in $I_j$.
$\forall x \in I_j$, such that $x \notin L_j$, $x$ is not a covered item, then $x$ is denoted as a \emph{covering item}. Obviously, the set of covering items is $I_j \setminus L_j$.
For instance, in Figure~\ref{fig:covpb:disassosiated},  $I_1$ = \{$a, b, c, d$\}. 
The support of the itemset $I_1$ in $R_{C_1}$, which is $s(I_1, R_{C_1})$, is equal to 4. In turn, 
it is equal to the minimum support of the items in $I_1$, which, in our example, is 
$s(\{c\}, R_{C_1}) = 4$.
Therefore, we say that the item $c$ is covered by the items $a$ and $b$. 
Similarly, the item  $d$ is also covered by the items $a$ and $b$.
%In other words in $\mathcal{T}^*$, a cover problem is practically detected for Chicago city since it was present exclusively in combination with all the other cities of $R_{C_{1}}$.

% interprétation en termes de capteurs de localisation 

Intuitively, a privacy breach occurs if an attacker is able to link $m$ items from his/her background knowledge, to less than $k$ records in all the datasets reconstructed by the inverse transformation. More subtle is when these records contain the same set of items in all the reconstructed datasets, thus linking more than $m$ items to the individual, or worse leading to a complete de-anonymization by linking, with certainty, the complete set of items to the individual.

We will show in the following that this privacy breach might occur whenever the dataset is subject to a cover problem.
Formally speaking:
\begin{lemma}\label{lemma:pb}
	Let $\mathcal{T}^*$ be a $k^m$-disassociated dataset subject to a cover problem. The disassociation guarantee is thus not valid for $m \geq 2$.
\end{lemma}
%We show in the proof that an attacker can breach the privacy resulting from the disassociation technique.

\begin{proof} 
Let $\mathcal{T}^*$ be a $k^m$-disassociated dataset subject to a cover problem.
The following set 
$I_j = 
	\{x | x \in  R_{C_{j}}\}
	\label{eq:Ij}$
	 is thus not empty and there exists a covered item $z\in I_j$ such that 
	 $s(I_j, R_{C_{j}})  = s(\{z\}, R_{C_{j}})$.
	 This means that each record $r$ of $R_{C_{j}}$ that contains $z$ includes also $I_j$. 
%Suppose now that the attacker knows about an individual the set $\{z,y \}$ where 
Suppose now that the attacker's background knowledge is the set $\{z,y \}$ where 
$y$ is an item in another record chunk $R_{C_{l}}$. 

By contradiction, suppose that the disassociation guarantee is valid, \textit{i.e.}, 
$z$ and $y$ are associated together in $k$ records in at least one of the datasets 
$\mathcal{T}^{\textcircled{.}}$, reconstructed by the inverse transformation of $\mathcal{T}^*$.
Since $z$ is a covered item, it appears in each  record $r$ defined above.
The item $y$ will also be associated $k$ times with all the items in $I_j$. 

While this is correct from a privacy perspective, it cannot be considered for disassociation. 
Items, $y$, $z$ and any covering item $x \in I_j$ are indeed 
considered as $k^m$-anonymous, and, therefore, should have been allot to the same record chunk 
$R_{C_{j}}$ according to disassociation\footnote{Vertical partitioning creates $k^m$-anonymous record chunks.}, 
whereas $z$ and $y$ are respectively items of chunks $R_{C_{j}}$ and $R_{C_{l}}$  by hypothesis.

%Now, if the attacker's background knowledge consists of itemsets of size $m$ and there exists an itemset $I$ that contains both $x_{i,j}$ and $y_{i,j-1}$, 
%( \textit{i.e.}, $\{x_{i,j}, y_{i,j-1}\} \subseteq I$) a privacy breach will occur. In fact, these $m$ items will never be associated together in $k$ records in any of the datasets reconstructed by the inverse transformation of $\mathcal{T}^*$.\vspace{-7pt}
\end{proof}

	\section{Safe Disassociation}\label{sec:PPPS}
	In this section, we show that a safe disassociation can be achieved, to ensure that a released/published dataset is no longer subject to a cover problem. This privacy guarantee is formally defined as follows: 
     
\begin{definition}[Safe Disassociation]\label{th:safe:disassociation}
		Let $\mathcal{G}$ be the inverse transformation of $\mathcal{T}^*$ with 
		respect to a $k^m$-disassociation, whose set of items is $\mathcal{D}$. 
		The dataset $\mathcal{T}^*$ is \emph{safely disassociated} if 
		$\forall I \subseteq \mathcal{D}$ such that
		$|I| \leq m$, there exists 
		
		$\mathcal{T}^{\textcircled{.}} \in \mathcal{G}(\mathcal{T}^*) $ with  
		$s(I, \mathcal{T}^{\textcircled{.}}) \geq k$ and 
		$\mathcal{T}^{\textcircled{.}}$ is not subject to a cover problem.
\end{definition}

Safe disassociation ensures that at least a dataset $\mathcal{T}^{\textcircled{.}}$ 
reconstructed by the inverse transformation of $\mathcal{T}^*$:
\begin{itemize}
    \item contains $k$ records for an itemset of size $m$ or less, abiding to the ${k}^{m}$-anonymity privacy constraint, and 
    \item the dataset $\mathcal{T}^{\textcircled{.}}$ has no covered items. 
\end{itemize}

In the following, we show how this safe disassociation can be achieved by applying a partial suppression on a disassociated dataset.

\subsection{Achieving safe disassociation with partial suppression}

In previous works dedicated to privacy preservation~\cite{Jia2014}, partial suppression is used to ensure that no sensitive rules can be inferred with a confidence greater than a certain threshold. Here, we assume that partial suppression can achieve disassociation safely regardless the sensitivity of the data;
all items are considered with the same level of sensitivity. Moreover, using partial suppression, we minimize the information loss with respect to our privacy guarantee. Unlike global suppression that removes the items in question from all the records, partial suppression remains more efficient in terms of utility.
\begin{partialsuppression}
Applying the following rules until the hypothesis is established produces a safely disassociated dataset.  
\begin{itemize}	
\item \emph{Hypothesis:} Let $I$, as defined in Definition~\ref{def:coverpb}, be the itemset of the record chunk $R_{C_j}$ that suffers from a cover problem. 
\item \emph{Preconditions:} Let $\textit{card} = \lceil \dfrac{|I|}{2} \rceil$ be the count of records that are going to be partially suppressed from  $R_{C_j}$, and let $\delta$ be the maximum number of records allowed in the cluster. Partial suppression is applicable over $R_{C_j}$ when:
    \begin{eqnarray}
    |R_{C_j}| &\leq &  \delta -2 \label{eq:pre:preservingmaxclustersize}\\
    s(I, R_{C_j}) &\geq &k+\min(\textit{card},m) \label{eq:pre:preservingkmanonymity}
    \end{eqnarray}
    
\item \emph{Rules:} 
\begin{enumerate}
    \item Create a random partition $I_{1} \cup I_{2} .... \cup I_{\textit{card}}$ of $I$ where each set $I_j$ is composed of two items, but $I_{\textit{card}}$, which is a singleton if the cardinality of $I$ is odd. 
    \item Create two empty sets $L_{1}$ and $L_{2}$, known as the ghost records.
    \item For each $I_{i} \in I$, successively:
	    \begin{enumerate}
	    \item\label{item:partialsup:suppr} suppress  $I_{i}$ from a record $r_i \in R_{C}$ such that $r_{i} = I$ and 
	    \item add the items   $x_{1}$ and $x_{2}$ from $I_{i}$ to respectively $L_{1}$ and $L_{2}$:\\
         $L_{1} = L_{1} \cup \{x_{1}\} $ and  $L_{2} = L_{2} \cup \{x_{2}\} $
        \end{enumerate}
	\item Add the two ghost records $L_{1}$ and $L_{2}$ to $R_{C_j}$ \\
\end{enumerate}
\end{itemize}
\end{partialsuppression}
%Intuitively, we pinpoint the covered items in a record chunk  \textit{i.e.}, $L_{i,j-1} = \{p \in  R_{i_{C_{j-1}}} | s(p, R_{i_{C_{j-1}}}) = s(I_{i,j-1}, R_{i_{C_{j-1}}}) \}$, where $I_{i,j-1}$ is  defined as in Definition~\ref{def:coverpb}. As preconditions to partial suppression, the support of $I_{i,j-1}$ must be greater than the \textit{card}inality of $L_{i,j-1}$ and greater or equal to $k + m$.\\
%Practically, we suppress first all the covering items $L_{i,j-1} \backslash I_{i,j-1}$ from a record $r \subseteq R_{i_{C_{j-1}}}$, 
%such that $I_{i,j-1} \subseteq r$. Then for each covered item $p \in L_{i,j-1}$ we find a record $r' \subseteq R_{i_{C_{j-1}}}$, such that $I_{i,j-1} \subseteq r'$ and we suppress $p$ from it.
By definition, a \emph{cover problem} is the ability to link one-to-one or one-to-many items in two record chunks. 
This arises when there exists $x \in I_j$ such that $s(I_j, R_{C_{j}})  = s(\{x\}, R_{C_{j}})$. 
The aim of partial suppression is to suppress items in such a way as to ensure that $s(I_j, R_{C_{j}})$ remains different than $s(\{x\}, R_{C_{j}})$, thus $s(I_j, R_{C_{j}}) \neq s(\{x\}, R_{C_{j}})$.

\subsubsection{Discussion on preconditions}
Preconditions \eqref{eq:pre:preservingmaxclustersize} and \eqref{eq:pre:preservingkmanonymity} play an important role in ensuring that the safe disassociation is preserved throughout the process of partial suppression. 
\begin{description}
	\item[Precondition \eqref{eq:pre:preservingmaxclustersize}] is defined to keep the size of the clusters bounded by the maximum cluster size $\delta$. In fact, in the horizontal partitioning, no additional records are added to the cluster if the maximum cluster size is reached. As a consequence, the record chunks created from the vertical partitioning  have a cardinality less or equal to $\delta$. Hence, to add the two ghost records $L_{1}$ and $L_{2}$ to a record chunk, its cardinality must be less than $\delta -2$.
	
	\item[Precondition \eqref{eq:pre:preservingkmanonymity}]  is defined to preserve the $k^m$-anonymity constraint in the record chunk. 
	Partial suppression modifies $\textit{card}$ occurrences of $I$, and, therefore, leaves $s(I,R_{C_j}) -\textit{card}$ unmodified. Thus,
	we have:
	\begin{eqnarray}
	s(I,R'_{C_j})   & =     &  s(I,R_{C_j}) -\textit{card} \nonumber\\
	& \geq  & k + \min(m, \textit{card}) -\textit{card} \label{eq:majorationSIR'}
	\end{eqnarray}
	
	Let $X=\{x_{1}, x_{2}, ... x_{m}\}$ be a subset of $I$, $X\subset I$. There are two cases to consider, depending on whether $m$ is greater than $\textit{card}$ or not. 
	\begin{itemize}
		\item if $\textit{card} \leq m$: $s(I,R'_{C_j})$ is greater than $k$ according to inequality~\eqref{eq:majorationSIR'}. 
		We have $s(X,R'_{C_j}) \geq s(I,R'_{C_j}) \geq k$. This means that $k^m$-anonymity is satisfied.
		
		\item if $m < \textit{card}$: %and let us focus on 
		%counting the number of occurrences of any set $\{x_{1}, x_{2}, ... x_{m}\}$ in the modified record chunk $R'_{C_j}$.
		%It first appears in the  
		$s(I,R_{C_j}) -\textit{card}$ records are kept unchanged by the partial suppression. After applying step~\ref{item:partialsup:suppr} in partial suppression, we notice that $s(X,R'_{C_j}) \geq \textit{card} -m$. This is because the smallest value, $\textit{card} -m$, is obtained when each item in $X$ belongs to a distinct pair $I_j$.
		In this situation, $m$ records do not associate items $x_{1}$, $x_{2}$,\ldots,$x_{m}$ together.
		As a result, the following applies according to hypothesis~\eqref{eq:pre:preservingkmanonymity}.
		\begin{eqnarray}
	s(X,R'_{C_j}) &\geq & s(I,R_{C_j}) -\textit{card} + \textit{card} -m  \nonumber\\
	&\geq & s(I,R_{C_j}) -m  \nonumber\\
	&\geq & k+\min(\textit{card},m) -m  \nonumber\\
	&\geq & k \nonumber
	\end{eqnarray}

	\end{itemize}
\end{description}

\begin{lemma}
Given a record chunk $R_{C_j}$, we say that partial suppression achieves safe disassociation on $R_{C_j}$ if preconditions ~\eqref{eq:pre:preservingmaxclustersize} and~\eqref{eq:pre:preservingkmanonymity} are verified.
\end{lemma}

\begin{proof}

Let $R'_{C_j}$ be the result of applying the partial suppression rules on the record chunk $R_{C_j}$. 
The itemset $I$ is randomly partitioned in $\textit{card}$ (where $\textit{card} = \lceil \dfrac{|I|}{2} \rceil$) disjoint subsets of cardinality equal to 2, and possibly one singleton if $|I|$ is odd. Those subsets are used successively to perform partial suppression. 
Since the preconditions are verified, it is straightforward to prove that due to the rules of partial suppression, the support of any item in $R_{C_j}$ remains the same. 
$\forall x \in I$, $x$ is suppressed from a record $r_{i}= I$ and then added randomly to one of the ghost records $L_{1}$ or $L_{2}$,  therefore:
\begin{equation}
s(\{x\}, R'_{C_j}) = s(\{x\}, R_{C_j}). \label{eq:suporxdoesnotchange}
\end{equation}

However, what has been partially suppressed is the association between any two items $x$ and $y$ in $I$. Two case scenarios can arise:
\begin{itemize}
\item items $x$ and $y$ belong to the same subset $I_{i}$ of $I$.
Thus, $x$ and $y$ are suppressed once from the same record $r_{i}$, and added to $L_{1}$ and $L_{2}$ respectively. This ensures that they cannot be associated again in the ghost records:
$$s(\{x,y\},R'_{C_j}) = s(\{x,y\},R_{C_j}) -1 =s(\{x\},R'_{C_j}) -1= s(\{y\},R'_{C_j})-1$$
\item 
items $x$ and $y$ belong to two different subsets of $I$, $I_i$ and $I_j$ respectively.
The association between $x$ and $y$ is lost twice since these items were suppressed from two distinct records in $R_{C_j}$, then:
$s(\{x,y\}, R'_{C_j})$ is equal to  $s(\{x,y\}, R_{C}) -2$ (resp. to $s(\{x,y\}, R_{C}) -1$)
if $x$ and $y$ are added to the different ghost records $L_{1}$ and $L_{2}$
(resp. if $x$ and $y$ are added to the same ghost records).
\end{itemize}
In both cases, we have:
\begin{eqnarray*}
s(\{x,y\},R'_{C_j}) &<& s(\{x\},R'_{C_j}) \\ 
                    &<&s(\{y\},R'_{C_j})
\end{eqnarray*}

Due to the inequality $S(I,R'_{C_j}) \le s(\{x,y\},R'_{C_j})$ and $\forall x, y$ in $I$ we have 
$ s(I,R'_{C}) < s(\{x\},R'_{C}),$
which concludes the proof.\\
\end{proof}

\begin{figure}[ht]
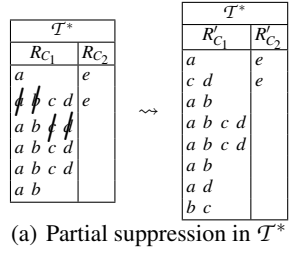
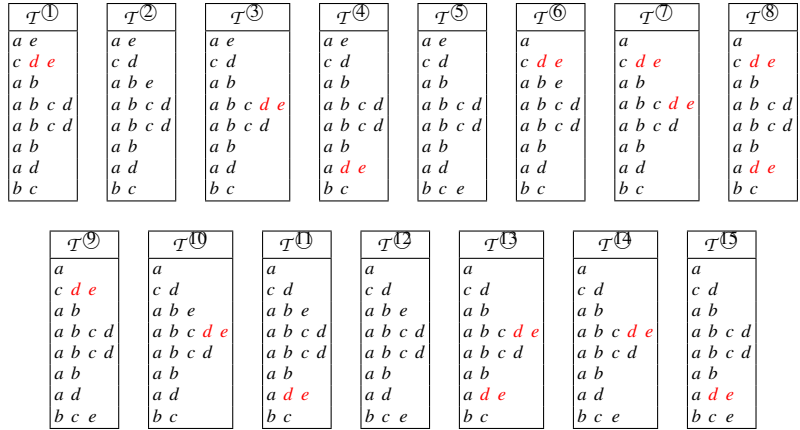

\begin{center}
\begin{scriptsize}
    \subfigure[Partial suppression in $\mathcal{T}^*$]{
    \begin{minipage}{1\textwidth}
    \begin{center}
	    \[
	   \arraycolsep=1.4pt
        \begin{array}{|llll|l|}
        \hline
        \multicolumn{5}{|c|}{\mathcal{T}^*}\\
            \hline
        \multicolumn{4}{|c|}{R_{C_1}} & R_{C_2}\\
        \hline
        a &   &   &   & e \\
        \cancel{a} & \cancel{b} & c & d & e \\
        a & b & \cancel{c} & \cancel{d} & \\
        a & b & c & d & \\
        a & b & c & d & \\
        a & b &   &   & \\
        \hline
        \end{array}
        \,\,\,\,\,\,\,\,\,
        \leadsto
        \,\,\,\,\,\,\,\,\,
        \begin{array}{|llll|l|}
        \hline
        \multicolumn{5}{|c|}{\mathcal{T}^*}\\
            \hline
        \multicolumn{4}{|c|}{R'_{C_1}} & R'_{C_2}\\
        \hline
        a &   &   &   & e \\
        c & d &   &   & e \\
        a & b &   &   & \\
        a & b & c & d & \\
        a & b & c & d & \\
        a & b &   &   & \\
        a & d &   &   & \\
        b & c &   &   & \\
        \hline 
        \end{array}
        \]  
        \end{center}
      \end{minipage}
        \label{fig:partsuppex}
    }   
    
    \subfigure[Inverse transformation after applying partial suppression on $\mathcal{T}^*$]{
      \begin{minipage}{1\textwidth}
        \begin{center}
	    \[
	    \arraycolsep=1.4pt
        \begin{array}{|llll|}
            \hline
            \multicolumn{4}{|c|}{\mathcal{T}^{\textcircled{1}}}\\
            \hline
            a & e &   &\\
            c & {\color{red}d} & {\color{red}e} &\\
            a & b &   &\\
            a & b & c & d \\
            a & b & c & d \\
            a & b &   &\\
            a & d &   &\\
            b & c &   &\\
            \hline
        \end{array}        
        \,\,\,\,\,\,\,\,\,\,\,\,\,\,
        \begin{array}{|llll|}
	        \hline
	        \multicolumn{4}{|c|}{\mathcal{T}^{\textcircled{2}}}\\
	        \hline
	        a & e &   & \\
	        c & d &   & \\
	        a & b & e & \\
	        a & b & c & d \\
	        a & b & c & d \\
	        a & b &   & \\
	        a & d &   & \\
	        b & c &   & \\
	        \hline
        \end{array}        
        \,\,\,\,\,\,\,\,\,\,\,\,\,\,
        \begin{array}{|lllll|}
        \hline
        \multicolumn{5}{|c|}{\mathcal{T}^{\textcircled{3}}}\\
	        \hline
	        a & e &   &   & \\
	        c & d &   &   & \\
	        a & b &   &   & \\
	        a & b & c & {\color{red}d} & {\color{red}e} \\
	        a & b & c & d & \\
	        a & b &   &   & \\
	        a & d &   &   & \\
	        b & c &   &   & \\
	        \hline
        \end{array}        
        \,\,\,\,\,\,\,\,\,\,\,\,\,\,
        \begin{array}{|llll|}
	        \hline
	        \multicolumn{4}{|c|}{\mathcal{T}^{\textcircled{4}}}\\
	        \hline
	        a & e &   & \\
	        c & d &   & \\
	        a & b &   & \\
	        a & b & c & d \\
	        a & b & c & d  \\
	        a & b &   & \\
	        a & {\color{red}d} & {\color{red}e} & \\
	        b & c &   & \\
	        \hline
        \end{array}       
        \,\,\,\,\,\,\,\,\,\,\,\,\,\,
		\begin{array}{|llll|}
			\hline
			\multicolumn{4}{|c|}{\mathcal{T}^{\textcircled{5}}}\\
			\hline
			a & e &   & \\
			c & d &   & \\
			a & b &   & \\
			a & b & c & d \\
			a & b & c & d \\
			a & b &   & \\
			a & d &   & \\
			b & c & e & \\
			\hline
		\end{array}        
        \,\,\,\,\,\,\,\,\,\,\,\,\,\,
		\begin{array}{|llll|}
			\hline
			\multicolumn{4}{|c|}{\mathcal{T}^{\textcircled{6}}}\\
			\hline
			a &   &   & \\
			c & {\color{red}d} & {\color{red}e} & \\
			a & b & e & \\
			a & b & c & d \\
			a & b & c & d \\
			a & b &   & \\
			a & d &   & \\
			b & c &   & \\
			\hline
		\end{array}  
		\,\,\,\,\,\,\,\,\,\,\,\,\,\,
		\begin{array}{|lllll|}
			\hline
			\multicolumn{5}{|c|}{\mathcal{T}^{\textcircled{7}}}\\
			\hline
			a &   &   &   & \\
			c & {\color{red}d} & {\color{red}e} &   & \\
			a & b &   &   & \\
			a & b & c & {\color{red}d} & {\color{red}e} \\
			a & b & c & d & \\
			a & b &   &   & \\
			a & d &   &   & \\
			b & c &   &   & \\
			\hline
		\end{array}
		\,\,\,\,\,\,\,\,\,\,\,\,\,\,
		\begin{array}{|llll|}
		\hline
		\multicolumn{4}{|c|}{\mathcal{T}^{\textcircled{8}}}\\
		\hline
		a &   &   & \\
		c & {\color{red}d} & {\color{red}e} & \\
		a & b &   & \\
		a & b & c & d \\
		a & b & c & d \\
		a & b &   & \\
		a & {\color{red}d} & {\color{red}e} & \\
		b & c &   & \\
		\hline
		\end{array}
        \]
        
        \noindent
        \[
        \arraycolsep=1.4pt
        \begin{array}{|llll|}
        \hline
        \multicolumn{4}{|c|}{\mathcal{T}^{\textcircled{9}}}\\
        \hline
        a &   &   & \\
        c & {\color{red}d} & {\color{red}e} & \\
        a & b &   & \\
        a & b & c & d \\
        a & b & c & d \\
        a & b &   & \\
        a & d &   & \\
        b & c & e & \\
        \hline
        \end{array}
        \,\,\,\,\,\,\,\,\,\,\,\,\,\,
        \begin{array}{|lllll|}
	        \hline
	        \multicolumn{5}{|c|}{\mathcal{T}^{\textcircled{10}}}\\
	        \hline
	        a &   &   &   & \\
	        c & d &   &   & \\
	        a & b & e &   & \\
	        a & b & c & {\color{red}d} & {\color{red}e} \\
	        a & b & c & d & \\
	        a & b &   &   & \\
	        a & d &   &   & \\
	        b & c &   &   & \\
	        \hline
        \end{array}        
        \,\,\,\,\,\,\,\,\,\,\,\,\,\,
		\begin{array}{|llll|}
			\hline
			\multicolumn{4}{|c|}{\mathcal{T}^{\textcircled{11}}}\\
			\hline
			a &   &   & \\
			c & d &   & \\
			a & b & e & \\
			a & b & c & d \\
			a & b & c & d \\
			a & b &   & \\
			a & {\color{red}d} & {\color{red}e} & \\
			b & c &   & \\
			\hline
		\end{array}       
        \,\,\,\,\,\,\,\,\,\,\,\,\,\,
		\begin{array}{|llll|}
			\hline
			\multicolumn{4}{|c|}{\mathcal{T}^{\textcircled{12}}}\\
			\hline
			a &   &   & \\
			c & d &   & \\
			a & b & e & \\
			a & b & c & d \\
			a & b & c & d \\
			a & b &   & \\
			a & d &   & \\
			b & c & e & \\
			\hline
		\end{array}        
        \,\,\,\,\,\,\,\,\,\,\,\,\,\,
		\begin{array}{|lllll|}
			\hline
			\multicolumn{5}{|c|}{\mathcal{T}^{\textcircled{13}}}\\
			\hline
			a &   &   &   & \\
			c & d &   &   & \\
			a & b &   &   & \\
			a & b & c & {\color{red}d} & {\color{red}e} \\
			a & b & c & d & \\
			a & b &   &   & \\
			a & {\color{red}d} & {\color{red}e} &   & \\
			b & c &   &   & \\
			\hline
		\end{array}         
        \,\,\,\,\,\,\,\,\,\,\,\,\,\,
		\begin{array}{|lllll|}
			\hline
			\multicolumn{5}{|c|}{\mathcal{T}^{\textcircled{14}}}\\
			\hline
			a &   &   &   & \\
			c & d &   &   & \\
			a & b &   &   & \\
			a & b & c & {\color{red}d} & {\color{red}e} \\
			a & b & c & d & \\
			a & b &   &   & \\
			a & d &   &   & \\
			b & c & e &   & \\
			\hline
		\end{array} 
		        \,\,\,\,\,\,\,\,\,\,\,\,\,\,
		\begin{array}{|llll|}
		\hline
		\multicolumn{4}{|c|}{\mathcal{T}^{\textcircled{15}}}\\
		\hline
		a &   &   & \\
		c & d &   & \\
		a & b &   &  \\
		a & b & c & d \\
		a & b & c & d \\
		a & b &   & \\
		a & {\color{red}d} & {\color{red}e} & \\
		b & c & e & \\
		\hline
		\end{array}          
        \]
        \end{center}
      \end{minipage}
        \label{fig:covpb:preimages1}
    }
    \end{scriptsize}
  \end{center}
	\caption{Eliminating a cover problem with partial suppression}
	\label{fig:partialsupp}
\end{figure}

% Dans l'exemple suivant faire le lien avec les capteurs de localisation.

To illustrate how a cover problem can be eliminated through partial suppression, let us consider the example in Figure \ref{fig:partsuppex}. Both items $c$ and $d$ are covered items in the disassociated dataset $\mathcal{T}^*$. After applying partial suppression, in Figure \ref{fig:partsuppex}, two subsets $I_{1} = \{a,b\}$ and $I_{2}=\{c,d\}$ are created after randomly partitioning $I = \{a,b,c,d\}$. From $I_{1}$ and $I_{2}$, two ghost records are created $L_{1} =\{a,d\}$ and $L_{2}=\{b,c\}$, containing one of the items from each suppressed subsets. Next, these two ghost records are added to $R_{C_{1}}$. We illustrate, in Figure \ref{fig:covpb:preimages1}, all possible reconstructed datasets of the final disassociated dataset. Now, if an attacker knows that a specific individual has searched for items $\{d,e\}$ (considered as the attacker's background knowledge), he/she will be able to associate them with three possible records $\{c,d,e\}$, $\{a,d,e\}$ and $\{a,b,c,d,e\}$. While these extra associations are considered noise, for the sake of privacy, they added ambiguity to the result since the attacker cannot link $\{d,e\}$ to a particular record. %Contrary to the original result where the attacker knew that an individual searched for itemset $\{d,e\}$, he/she now cannot deduce with certainty the exact set of all items related to that individual.  
In the next section, we will study and evaluate the impact of partial suppression on the utility of the dataset.

	\section{Experimental Evaluation} \label{sec:exper}
	
In keeping with the previous work~\cite{DBLP:conf/secrypt/BarakatBNG16}, we elaborate our experiments on two datasets, the \textit{BMS}1 and the \textit{BMS}2, which contain click-stream E-commerce data. %This choice will reflect, through the experiments, the privacy breach due to disassociation in relatively small ($\textit{BMS}1$) and big ($\textit{BMS}2$) datasets, as well as the effect of our algorithm on the them. \\
Table \ref{tab:datasets} shows the properties of the datasets.

The aim of the experiments can be summarized as follows:
\begin{itemize}
	\item Evaluating the privacy breach in the disassociated dataset.
	\item Evaluating the utility loss w.r.t the suppression of items (or their occurrences).
	\item Studing the loss in associations when disassociating and safely disassociating a dataset. 
	\item Evaluating the performance of partial suppression.
\end{itemize}
 
\begin{table}
	\centering
	\small
	\caption{Datasets properties}
	\begin{tabular}{|l|l|l|l|p{3cm}|}\hline
		Dataset & \# of distinct individuals & \# of distinct items & count of items' occurrences\\ \hline
		\textit{BMS}1 & 59602 & 497 & 149639 \\ \hline
		\textit{BMS}2 & 77512 & 3340 & 358278 \\ \hline
	\end{tabular}
	\label{tab:datasets}
\end{table}

\subsection{Privacy and utility metrics}

\subsubsection{Privacy Evaluation Metric (PEM)} 
represents the number of vulnerable record chunks in a disassociated dataset. In fact, we consider that every record chunk that is subject to a cover problem is a vulnerable record chunk. We formally define our $PEM$ as:
\begin{center}
$PEM =\frac {vRC}{RC}$
\end{center}

where,
 \begin{itemize}
 	\item $vRC$ represents the number of vulnerable record chunks in $\mathcal{T}$, the disassociated dataset, and 
 	\item $RC$ represents the total number of record chunks in $\mathcal{T}$.
 \end{itemize}

\subsubsection{Relative Loss Metric (RLM)}
determines the relative number of suppressed occurrences of items with partial suppression. Formally,
\begin{displaymath}
RLM =\frac {\sum_{\forall x \in \mathcal{D} } ( {s(\{x\},\mathcal{T}^*) - s(\{x\},\mathcal{T}' )} )}{\sum_{\forall x \in \mathcal{D} } ({s(\{x\},\mathcal{T}^*)})}
\end{displaymath}
where, 
\begin{itemize}
	\item  $s(\{x\},\mathcal{T}^*)$ represents the support of the item $x$ in the disassociated dataset $\mathcal{T}^*$, and 
	\item  $s(\{x\},\mathcal{T}')$ represents the support of the item $x$ in the safely disassociated dataset  $\mathcal{T}'$.
\end{itemize}

\subsubsection{Relative Association Error (RAE)}
evaluates how likely two items remain associated together in an anonymized dataset \cite{terrovitis2012}. 
In fact, using $RAE$, we are able to evaluate the information loss due to the anonymization of the dataset (whether it is disassociated or safely disassociated). Formally, $RAE$ is defined as follows:

\[RAE =  \frac {s(\{x,y\},\mathcal{T}) - s(\{x,y\},[\mathcal{T}^*|\mathcal{T}'] )}{AVG(s(\{x,y\},\mathcal{T}),s(\{x,y\},[\mathcal{T}^*|\mathcal{T}'] ))}\]
where, 
\begin{itemize}
	\item $s(\{x,y\},\mathcal{T})$ represents the support of items $\{x,y\}$ in the original dataset $\mathcal{T}$, the disassociated dataset $\mathcal{T}^*$, or the safely disassociated dataset $\mathcal{T}'$.
\end{itemize}

\subsection{Experimental results} \label{sec:experimentalresults}
In this section, we present the results of the conducted experiments, evaluating privacy and data utility over the safely disassociated dataset.\\
%Throughout the experiments, we vary the maximum cluster size $\delta$ between 10 and 60, increased by 10, to evaluate different disassociations for the same dataset.
\subsubsection{Evaluating the privacy breach}
We consider that a privacy breach occurs if an attacker is able to link $m$ items, which he/she already knows about an individual, to less than $k$ records in all the datasets reconstructed by the inverse transformation. In this test, we study the impact of the cover problem on the privacy of the disassociated dataset. In fact,  we consider that a potential privacy breach exists whenever a cover problem is identified in a record chunk regardless of the background knowledge of the attacker. It is typically a strong attacker \cite{DBLP:conf/secrypt/BarakatBNG16} who is able to link any two items to a specific individual. Hence, we determine the following: 
\begin{itemize}
	\item the relationship between the $PEM$ and both, $k$ and $m$.
	\item the relationship between the $PEM$ and the maximum cluster size $\delta$.
\end{itemize}

\begin{description}
	\item[Varying $k$ and $m$: ] we vary $k$ and $m$ from 2 to 6. For each value, we compute the $PEM$ to evaluate how the privacy constraint remains satisfied in the record chunks. Figure~\ref{fig:PBeval1} shows the results of the evaluation. When $k$ increases in both datasets $\textit{BMS}1$ and $\textit{BMS}2$, the $PEM$ decreases from $46\%$ to $30\%$ in $\textit{BMS}1$ and from $56\%$ to $39\%$ in $\textit{BMS}2$. While varying $k$ affects the $PEM$, varying $m$ has no noticeable impact on the $PEM$ in both datasets. This is not surprising since, due to the cover problem, any two items known to the attacker can lead to a privacy breach. 
	\\
	\begin{figure}
		\centering
		\begin{tikzpicture}
		\begin{axis}[
		xlabel={$k$},
		ylabel={PEM (\%)},
		grid = major,
		grid style = {dashed, gray!30},
		width=6.5cm, height=3.5cm, 
		legend entries = {\textit{BMS}2,\textit{BMS}1},
		legend style={font=\tiny, at={(1.35,1)}},	  	
		]
		\addplot+[restrict expr to domain={\coordindex}{34:38},blue,mark options={fill=blue}] table [col sep=comma,x=k, y=perc RC affected]{resultsgraph};
		\addplot+[restrict expr to domain={\coordindex}{22:26},red,mark options={fill=red}] table [col sep=comma,x=k, y=perc RC affected]{resultsgraph};
		\end{axis}
		\label{fig:PLMeval:1}
		\end{tikzpicture}
		\centering
		\begin{tikzpicture}
		\begin{axis}[
		xlabel={$m$},
		ylabel={PEM (\%)},
		grid = major,
		grid style = {dashed, gray!30},
		width=6.5cm, height=3.5cm, 
		legend entries = {\textit{BMS}2,\textit{BMS}1},
		legend style={font=\tiny, at={(1.35,1)}},	  	
		]
		\addplot+[restrict expr to domain={\coordindex}{39:45},blue,mark options={fill=blue}] table [col sep=comma,x=m, y=perc RC affected]{resultsgraph};
		\addplot+[restrict expr to domain={\coordindex}{27:31},red,mark options={fill=red}] table [col sep=comma,x=m, y=perc RC affected]{resultsgraph};
		\end{axis}
		\label{fig:PLMeval:2}
		\end{tikzpicture}
		\caption{Evaluating the $PEM$ while varying $k$ and $m$} 
		\label{fig:PBeval1} 
	\end{figure}
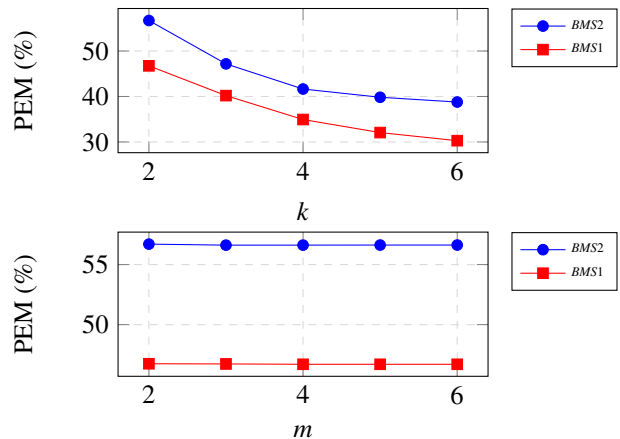
\item[Varying $\delta$: ] 
we vary $\delta$ from $10$ to $60$. For each value, we compute the $PEM$ to evaluate how the privacy constraint remains satisfied in the record chunks. The results in Figure \ref{fig:PBeval} show that the $PEM$ increases from $17\%$ to $52\%$ in $BMS1$ and from $26\%$ to $57\%$ in $BMS2$. This means that with more records in the cluster, the higher the chances are to compromise the dataset due to the cover problem. 
% results reflect the seriousness of the privacy breach occurring in a disassociated dataset due to the cover problem, where a high percentage of the record chunks are vulnerable, which puts associations that are intended to be hidden through the disassociation in danger. We notice that the the privacy loss metric $PLM$ is always higher in $\textit{BMS}2$ compared to $\textit{BMS}1$ through the different $\delta$, which means that the probability of finding vulnerable record chunks $vRC$ increases with the size of the datasets.
%On another hand, for both datasets, $vRC$ increases roughly with the increase of the maximum cluster size, $\delta$, from 10 to 60. This means that the ability to add more records to a cluster, would result in less record chunks achieving easier the $k^m-anonymity$ between associations by finding them more frequent in the cluster in question, putting the cover problem right up front. Therefore, better privacy preservation can be achieved with lower values of the predefined maximum cluster size.
%From these facts we can infer that treating the cover problem in any type of dataset is crucial. In the next experiments, we chose to vary $\delta$ to highlight its impact on that disassociation and analyze the results.
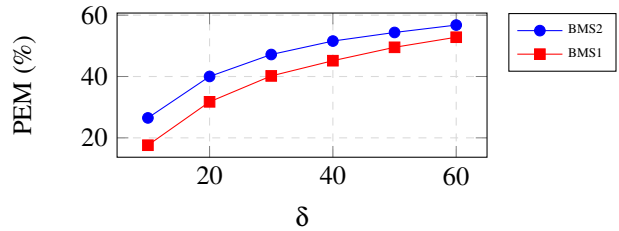
\begin{figure}[H]
	\centering
	\begin{tikzpicture}
	\begin{axis}[
	xlabel={$\delta$},
	ylabel={PEM (\%)},
	grid = major,
	grid style = {dashed, gray!30}, 
	width=6.5cm, height=3.5cm,  
	legend entries = {\textit \textit{BMS}2,\textit \textit{BMS}1},
	legend style={font=\tiny, at={(1.35,1)}},	
	]
	\addplot+[restrict expr to domain={\coordindex}{7:14},blue,mark options={fill=blue}] table [col sep=comma,x=MaxClusterSize, y=perc RC affected]{resultsgraph};
	\addplot+[restrict expr to domain={\coordindex}{0:5},red,mark options={fill=red}] table [col sep=comma,x=MaxClusterSize, y=perc RC affected]{resultsgraph};
	\end{axis}
	\label{fig:PLMeval:3}
	\end{tikzpicture}
	\caption{Evaluating the $PEM$ while varying $\delta$} 
	\label{fig:PBeval} 
\end{figure}

\end{description}
Now, given that the privacy breach is directly related to the value of $\delta$, we choose to vary $\delta$ and fix the values of $k$ and $m$ in the remaining tests. We use $k=3$ to keep computational time to a minimum and $m=2$ since only two items are sufficient to raise a privacy breach.
\subsubsection{Evaluating the number of suppressed items} 
In this test, we evaluate the number of items to be suppressed to safely disassociate a dataset. We note that the record chunks that do not comply with preconditions  \eqref{eq:pre:preservingmaxclustersize} and \eqref{eq:pre:preservingkmanonymity} in partial suppression, for the sake of privacy, their items are completely suppressed. We use the $RLM$ to estimate this loss of items. We vary the maximum cluster size $\delta$ from $10$ to $60$ and compute, for each value, the $RLM$. Figure \ref{fig:ILeval} shows the results of our evaluation. It is not surprising that $\textit{BMS}2$ is more susceptible to partial suppression since it is more vulnerable to the cover problem due to its size, as noted in the previous test (Figure \ref{fig:PBeval}).
In addition, in both datasets, when the maximum cluster size $\delta$ increases, we notice an increase in the $RLM$. This is actually consistent with the fact that the more vulnerable the record chunks are, the more items will be suppressed. Overall, the $20\%$ suppressions of items represent an acceptable trade-off between privacy and utility,  especially that, with safe disassociation, we release real datasets without modifying the items.
 
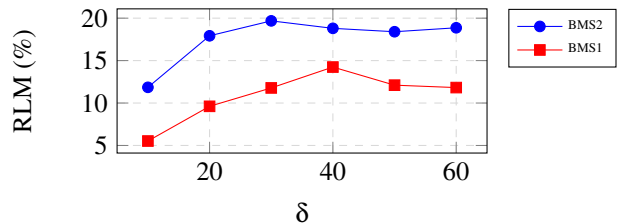
\begin{figure}[ht]
	\centering
	\begin{tikzpicture}
	\begin{axis}[
	xlabel={$\delta$},
	ylabel={RLM (\%)},
	grid = major,
	grid style = {dashed, gray!30}, 
	width=6.5cm, height=3.5cm,  
	legend entries = {\textit \textit{BMS}2,\textit \textit{BMS}1},
	legend style={font=\tiny, at={(1.35,1)}},	
	]
	\addplot+[restrict expr to domain={\coordindex}{7:14},blue,mark options={fill=blue}] table [col sep=comma,x=MaxClusterSize, y=percOccurencesLost]{resultsgraph};
	\addplot+[restrict expr to domain={\coordindex}{0:5},red,mark options={fill=red}] table [col sep=comma,x=MaxClusterSize, y=percOccurencesLost]{resultsgraph};
	\end{axis}
	\end{tikzpicture}
	\caption{Evaluating the $RLM$ while varying $\delta$} 
	\label{fig:ILeval} 
\end{figure}

\subsubsection{Evaluating the association error} 
In this test, we evaluate the error in associations, which is the result of dividing the records into chunks after the partitioning process. Eventually, some of the itemsets will be separated and their items will be stored into different chunks. This adds noise to the dataset since the support of associations between these separated items might be higher in the reconstructed datasets. This leads to this error in associations that can be calculated using the $RAE$, which is the relative difference between the support of the association of an itemset in the original dataset and the anonymized dataset \cite{anatomy}.
Again, we vary the maximum cluster size $\delta$ from $10$ to $60$ and compute, for each value, the $RAE$ on disassociated and safely disassociated datasets. The results in Figure \ref{fig:REeval} shows that the $RAE$ is higher in a safely disassociated dataset. It is not surprising because, with safe disassociation, we suppress some items to prevent the privacy breach. However, the difference between $RAE$ in safe disassociation and disassociation remains acceptable; varying between $1\%$ for $\textit{BMS}1$ and $0.1\%$ for $\textit{BMS}2$.

\begin{figure}[htt]
	\begin{tikzpicture}
	\begin{axis}[
	xlabel={$\delta$},
	ylabel={RAE},
	grid = major,
	grid style = {dashed, gray!30},
	width=6.5cm, height=3.5cm, 
	legend entries = {\textit{\textit{BMS}1} with disassociation,\textit{\textit{BMS}1} with safe disassociation},
	legend style={font=\tiny, at={(1.855,1)}},	  	
	]
	\addplot+[restrict expr to domain={\coordindex}{0:5},blue,mark options={fill=blue}] table [col sep=comma,x=MaxClusterSize, y=RE]{resultsgraph};
	\addplot+[restrict expr to domain={\coordindex}{0:5},red,mark options={fill=red}] table [col sep=comma,x=MaxClusterSize, y=RE Safe Diss]{resultsgraph};
	\end{axis}
	\end{tikzpicture}
	\centering
	\begin{tikzpicture}
	\begin{axis}[
	xlabel={$\delta$},
	ylabel={RAE},
	grid = major,
	grid style = {dashed, gray!30},
	width=6.5cm, height=3.5cm, 
	legend entries = {\textit{\textit{BMS}2} with disassociation,\textit{\textit{BMS}2} with safe disassociation},
	legend style={font=\tiny, at={(1.855,1)}},	  	
	]
	\addplot+[restrict expr to domain={\coordindex}{7:14},red,mark options={fill=red}] table [col sep=comma,x=MaxClusterSize, y=RE Safe Diss]{resultsgraph};
	\addplot+[restrict expr to domain={\coordindex}{7:14},blue,mark options={fill=blue}] table [col sep=comma,x=MaxClusterSize, y=RE]{resultsgraph};
	\end{axis}
	\label{fig:REeval:2}
	\end{tikzpicture}
	\caption{Evaluating the $RAE$ while varying $\delta$} 
	\label{fig:REeval} 
\end{figure}
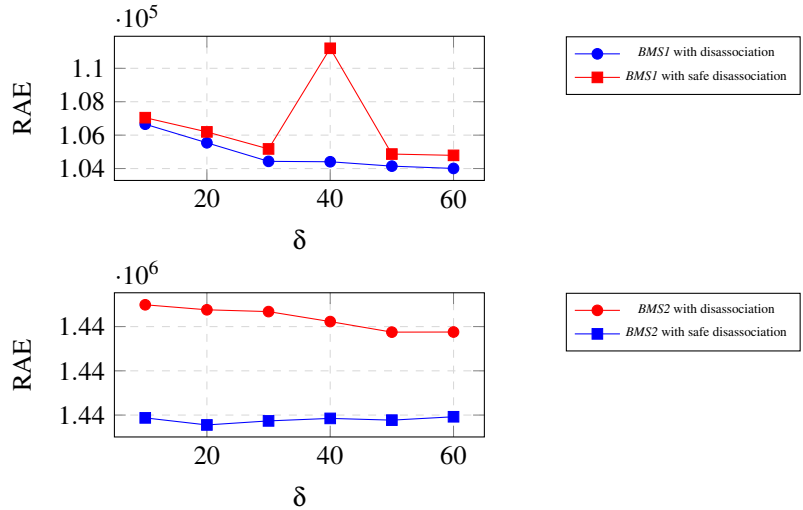

\subsubsection{Performance evaluation}
We compare the performance of our algorithm with different dataset size. The datasets used in this experiment are formed from \textit{BMS}1 and \textit{BMS}2. We vary, in the $x-axis$, the number of records with \{358k, 149K, 100K, 80K, 50k, 30k 10k\}; and the maximum cluster size $\delta$ is fixed to 30. Figure \ref{fig:perform} shows an increase in the run time with the increase of the dataset size.

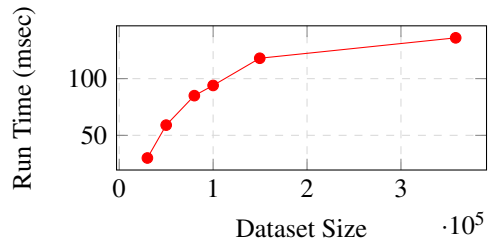
\begin{figure}[ht]
	\centering
	\begin{tikzpicture}
	\begin{axis}[
	xlabel={Dataset Size},
	ylabel={Run Time (msec)},
	grid = major,
	grid style = {dashed, gray!30},
	width=6.5cm, height=3.5cm,	  	
	]
	\addplot+[restrict expr to domain={\coordindex}{15:21},red,mark options={fill=red}] table [col sep=comma,x=DatasetSize, y=RunTime msec]{resultsgraph};
	\end{axis}
	\end{tikzpicture}
	\caption{Performance evaluation} 
	\label{fig:perform} 
\end{figure}

%To summarize, to keep an acceptable trade-off between privacy and utility, it is better to choose a low $\delta$ for disassociation. In addition and despite the high privacy breaches $PLM$ found and the information loss $ALM$, the relative association error $RAE$ were negligible for the different datasets. 

	\section{Conclusion} \label{sec:conclusion}
	
Disassociation is an interesting anonymization technique that is able to hide the link between individuals and their complete set of items while keeping the items without generalization. In a previous work \cite{DBLP:conf/secrypt/BarakatBNG16}, disassociation was considered vulnerable due to a cover problem. Basically, it is the ability to associate one-to-one or one-to-many items in two subsequent record chunks of the disassociated data. In this paper, we propose safe disassociation to solve this problem. We use partial suppression to achieve safe disassociation by suppressing some of the items that lead to a cover problem from subsequent record chunks. In the experiments, we show that the vulnerability of a disassociated dataset depends on the size of the cluster. We evaluate the utility of a safely disassociated dataset in terms of 1) the number of items to be suppressed to achieve safe disassociation, and 2) the additional noisy associations added due to item suppression and partitioning. The results of our evaluations showed that an acceptable trade-off between privacy and utility is met. 
%Mainly, partial suppression executes in two phases. First, from $I$ we construct sub-itemsets of cardinality equal to 2, and suppress them from distinct records representing $I$. Second, each item from the constructed sub-itemsets is added to one of the two ghost records, that are added to the record chunk being treated.
%Finally, we led a set of experiments, to test the privacy and the utility of the data after achieving the covering free disassociation. The experiments reflected interesting results. On the one hand, privacy is to be considered seriously after disassociation, where the ratio of vulnerable record chunks is high; but on the other hand, the proposed solution for the cover problem showed its efficiency on the level of data utility, where the association loss due to partial suppression is negligible. 

In future works, we aim at maximizing the utility of a safely disassociated dataset by modifying the clustering algorithm to keep user-defined itemsets associated together. 
	
	\section*{Acknowledgments}\label{sec:Acknowledgments}
	This work is funded by the InMobiles company\footnote{\url{www.inmobiles.net}} and the Labex ACTION program (contract ANR-11-LABX-01-01). Computations have been performed on the supercomputer facilities of the M\'{e}socentre de calcul de Franche-Comt\'e. Special thanks to  Ms. Sara Barakat for her contribution in identifying the cover problem.
\newcommand{\etalchar}[1]{$^{#1}$}

\bibliographystyle{alpha}
%\bibliography{references}

\end{document}